\documentclass{aamas2016}
\pdfoutput=1

\newif\ifshowtournaments
\showtournamentsfalse

\toappear{{\bf Appears in:} {\it \small Proceedings of the 15th International
		Conference on Autonomous Agents and Multiagent Systems (AAMAS 2016),\\
		J. Thangarajah, K. Tuyls,  C. Jonker, S. Marsella (eds.),\\May 9--13, 2016, Singapore.} 
	}

\usepackage{lmodern} 
\usepackage[table]{xcolor}
\usepackage[acm,amsthm]{pamas}
\usepackage[hyphens]{url}
\usepackage[hidelinks]{hyperref}
\hypersetup{breaklinks=true}
\usepackage{amsmath, amssymb, cleveref, booktabs, algorithm, graphicx, wrapfig}
\usepackage{lipsum}
\usepackage{fancyvrb} 
\usepackage{tikz}
\usetikzlibrary{shapes, calc, positioning}
\usepackage[noend]{algpseudocode}
\newcommand{\pref}{\succcurlyeq} 

\renewcommand*{\leq}{\leqslant}
\renewcommand*{\ge}{\geqslant}
\renewcommand*{\geq}{\geqslant}
\newcommand{\bigland}{\bigwedge} 
\newcommand{\biglor}{\bigvee} 
\newcommand{\formulaname}[1]{\textsf{#1}}

\renewcommand{\a}{\texttt a}
\renewcommand{\b}{\texttt b}
\renewcommand{\c}{\texttt c}
\renewcommand{\d}{\texttt d}
\newcommand{\x}{\mathtt x}
\newcommand{\y}{\mathtt y}

\newcommand{\electorates}{{\mathcal E(\mathcal N)}}

\newcommand{\yes}[1]{{\color{black!70!green}{\fontseries{b}\selectfont#1}}}
\newcommand{\no}[1]{{\color{black!70}\fontseries{l}\selectfont #1}}

\newcommand{\tc}[1][]{\ifthenelse{\equal{#1}{}}{\mathit{TC}}{\mathit{TC}(#1)}}

\newcounter{remark}

\renewcommand{\paragraph}[1]{\par\medskip\noindent\textbf{#1}\quad}

\newcolumntype{P}[1]{>{\centering\arraybackslash}p{#1}}

\newcommand{\supp}{\operatorname{supp}}

\begin{document}

\title{Optimal Bounds for the No-Show Paradox via SAT Solving}

\numberofauthors{3}

\author{
\alignauthor
 Felix Brandt\\
 	\affaddr{TU M\"unchen}\\
 	\affaddr{Munich, Germany}\\
 	\email{brandtf@in.tum.de}
 \alignauthor
 Christian Geist\\
 	\affaddr{TU M\"unchen}\\
 	\affaddr{Munich, Germany}\\
 	\email{geist@in.tum.de}
 \alignauthor
 Dominik Peters\\
 	\affaddr{University of Oxford}\\
     \affaddr{Oxford, UK}\\
      \email{dominik.peters@cs.ox.ac.uk}
}

\maketitle

\begin{abstract}
Voting rules allow multiple agents to aggregate their preferences in order to reach joint decisions. Perhaps one of the most important desirable properties in this context is \emph{Condorcet-consistency}, which requires that a voting rule should return an alternative that is preferred to any other alternative by some majority of voters. Another desirable property is \emph{participation}, which requires that no voter should be worse off by joining an electorate. A seminal result in social choice theory by \citet{Moul88b} has shown that Condorcet-consistency and participation are incompatible whenever there are at least~4 alternatives and~25 voters. We leverage SAT solving to obtain an elegant human-readable proof of Moulin's result that requires only~12 voters. Moreover, the SAT solver is able to construct a Condorcet-consistent voting rule that satisfies participation as well as a number of other desirable properties for up to 11~voters, proving the optimality of the above bound. We also obtain tight results for set-valued and probabilistic voting rules, which complement and significantly improve existing theorems.
\end{abstract}


\section{Introduction}

Whenever a group of autonomous software agents or robots needs to decide on a joint course of action in a fair and satisfactory way, they need to aggregate their preferences. A common way to achieve this is to use voting rules.
Voting rules are studied in detail in social choice theory and are coming under increasing scrutiny from computer scientists who are interested in their computational properties or want to utilize them in multiagent systems \citep[see, \eg][]{Roth15a,BCE+14a}.

In social choice theory, voting rules are usually compared using desirable properties (so-called axioms) that they may or may not satisfy. There are a number of well-known impossibility theorems---among which Arrow's impossibility is arguably the most famous---which state that certain axioms are incompatible with each other. These results, which show the non-existence of voting rules that satisfy a given set of axioms, are important because they clearly define the boundary of what can be achieved at all. 
This applies to the explicitly stated axioms as well as implicit ones such as boundaries on the number of voters or alternatives. For instance, if there are only two alternatives, Arrow's theorem does not apply and there are many voting rules, including majority rule, that satisfy the conditions used in Arrow's theorem.
One impossibility that requires unusually high bounds on the number of voters and alternatives is Moulin's \emph{no-show paradox} \citep{Moul88b}, which states that the axioms of Condorcet-consistency and participation are incompatible whenever there are at least 4 alternatives and 25 voters. 
Moulin proves that the bound on the number of alternatives is tight by showing that the maximin voting rule (with lexicographic tie-breaking) satisfies the desired properties when there are at most 3 alternatives. The tightness of the more restrictive condition on the number of voters was left open, however.
The goal of this paper is to give tight bounds on the number of voters required for Moulin's theorem and related theorems that appear in the literature. To achieve this, we encode these problems as formulas in propositional logic and then use SAT solvers to decide their satisfiability and extract minimal unsatisfiable sets (MUSes) in the case of unsatisfiability. This approach is based on previous work by \citet{TaLi09a}, \citet{GeEn11a}, \citet{BrGe15a}, and \citet{BBGH15a}. However, it turned out that a straightforward application of this methodology is insufficient to deal with the magnitude of the problems we considered. Several novel techniques were necessary to achieve our results. In particular, we extracted knowledge from computer-generated proofs of weaker statements and then used this information to guide the search for proofs of more general statements.

As mentioned above, Moulin's theorem uses the axioms of Condorcet-consistency and participation. Condorcet-consistency goes back to one of the most influential notions in social choice theory, namely that of a \emph{Condorcet winner}. A Condorcet winner is an alternative that is preferred to any other alternative by a majority of voters. The Marquis de Condorcet, after whom this concept is named, argued that, whenever a Condorcet winner exists, it should be elected \citep{Cond85a}. A voting rule satisfying this condition is called \emph{Condorcet-consistent}. Apart from the intuitive appeal of this condition, Condorcet-consistent rules are more robust to changes in the of feasible alternatives and less susceptible to strategic manipulation than other voting rules (such as Borda's rule) \citep[see, \eg][]{CaKe03a,DaMa08a}. 
While the desirability of Condorcet-consistency---as that of any other axiom---has been subject to criticism, many scholars agree that it is very appealing---if not indispensable---and a large part of the social choice literature deals exclusively with Condorcet-consistent voting rules \citep[\eg][]{Fish77a,Lasl97a,BCE+14a}.
\emph{Participation} was first considered by \citet{BrFi83a} and requires that no voter should be worse off by joining an electorate, or---alternatively---that no voter should benefit by abstaining from an election. The desirability of this axiom in any context with voluntary participation is evident. All the more surprisingly, \citeauthor{BrFi83a} have shown that single transferable vote (STV), a common voting rule, violates participation and referred to this phenomenon as the \emph{no-show paradox}. \citet{Moul88b}, perhaps even more startlingly, proved that no Condorcet-consistent voting rule satisfies participation when there are at least 25 voters. 

We leverage SAT solving to obtain an elegant human-readable proof of Moulin's result that requires only 12 voters. While computer-aided solving techniques allow us to tackle difficult combinatorial problems, they usually do not provide additional insight into these problems. Somewhat surprisingly, the computer-aided proofs we found possess a certain kind of symmetry that has not been exploited in previous proofs. 
Moreover, the SAT solver is able to construct a Condorcet-consistent voting rule that satisfies participation as well as a number of other desirable properties for up to 11 voters, proving the optimality of the above bound. This computer-generated voting rule is compatible with the maximin voting rule in 99.8\% of all cases and, in contrast to maximin, only selects alternatives from the top cycle.
As a practical consequence of our theorem, strategic abstention need not be a concern for Condorcet-consistent voting rules when there are at most 4 alternatives and 11 voters, for instance when voting in a committee.
We also use our techniques to provide optimal bounds for related results about set-valued and probabilistic voting rules \citep{JPG09a,Wood97a}. In particular, we give a tight bound of 17 voters for the optimistic preference extension, 14 voters for the pessimistic extension, and 12 voters for the stochastic dominance preference extension. These results are substantial improvements of previous results. For example, the previous statement for the pessimistic extension requires an additional axiom, at least 5 alternatives, and at least 971 voters \citep{JPG09a}.
Our results are summarized in \tabref{tbl:results}.

\begin{table*}[t]
	\colorlet{works}{green!35!black!25!white}
	\colorlet{paradox}{red!35!black!70!white}
	\setlength{\tabcolsep}{5.06pt}
	\begin{tabular}{lccccccccccccccccccccccccc}
		 \multicolumn{1}{r}{$n=$\!\!\!\!} & 1 & 2 & 3 & 4 & 5 & 6 & 7 & 8 & 9 & 10 & 11 & 12 & 13 & 14 & 15 & 16 & 17 & 18 & 19 & 20 & 21 & 22 & 23 & 24 & 25\\ \toprule
		 
		Condorcet & \multicolumn{11}{r|}{\cellcolor{works}\scriptsize Thm \ref{thm:positive-pairwise}\:$\rangle$\!\!\!}
		& \multicolumn{9}{l|}{\cellcolor{paradox}\color{white}\scriptsize \!\!\!\! $\langle$\:Thm \ref{thm:anonymous}} 
		& \multicolumn{4}{l|}{\cellcolor{paradox}\color{white}\scriptsize \!\!\!\!  $\langle$\:\citep{Kard14a}} 
		& \multicolumn{1}{l}{\cellcolor{paradox}\color{white}\scriptsize \!\!\!\!  $\langle$\:\citep{Moul88b}}\\
		
		Maximin & \multicolumn{6}{r}{\cellcolor{works}\scriptsize Thm \ref{thm:maximin}\:$\rangle$\!\!\!} & \multicolumn{19}{l}{\cellcolor{paradox}\color{white}\scriptsize \!\!\!\! $\langle$\:Thm \ref{thm:maximin}} \\
		
		Kemeny & \multicolumn{3}{r}{\cellcolor{works}\scriptsize Thm \ref{thm:kemeny}\:$\rangle$\!\!\!} & \multicolumn{22}{l}{\cellcolor{paradox}\color{white}\scriptsize \!\!\!\! $\langle$\:Thm \ref{thm:kemeny}} \\ \addlinespace[1pt]
		
		
		\midrule
		optimistic &  \multicolumn{16}{r}{\cellcolor{works}\scriptsize Thm \ref{thm:positive-optimist}\:$\rangle$\!\!\!}
		& \multicolumn{9}{l}{\cellcolor{paradox}\color{white}\scriptsize \!\!\!\! $\langle$\:Thm \ref{thm:optimist}} \\
		
		pessimistic & \multicolumn{13}{r}{\cellcolor{works}\scriptsize Thm \ref{thm:pessimist}\:$\rangle$\!\!\!} & \multicolumn{12}{l}{\cellcolor{paradox}\color{white}\scriptsize \!\!\!\! $\langle$\:Thm \ref{thm:pessimist}} \\ \addlinespace[1pt]
		
		
		\midrule
		strong SD &  \multicolumn{11}{r}{\cellcolor{works}\scriptsize Thm \ref{thm:sd-participation}\:$\rangle$\!\!\!}
		& \multicolumn{14}{l}{\cellcolor{paradox}\color{white}\scriptsize \!\!\!\! $\langle$\:Thm \ref{thm:sd-participation}} \\ \addlinespace[1pt]
		\bottomrule \\ \addlinespace[-7pt]
	\end{tabular}
	\setlength{\tabcolsep}{5.1pt}
	\begin{tabular}{lllll}
		\hspace{12cm} &
		\cellcolor{works}{\quad} & Possibility\: &
		\cellcolor{paradox}{\quad} & Impossibility
	\end{tabular}
	\vspace{-12pt}
	\caption{Bounds on the number of voters for which Condorcet extensions can satisfy participation. Green cells indicate the existence of a Condorcet extension satisfying participation (for $m=4$). Red cells indicate that no Condorcet extension satisfies participation (for $m\geq 4$).}
	\label{tbl:results}
\end{table*}

\section{Related Work}

The no-show paradox was first observed by \citet{BrFi83a} for the STV voting rule. \citet{Ray86a} and \citet{LeMe00a} investigate how frequently this phenomenon occurs in practice. The main theorem addressed in this paper is due to \citet{Moul88b} and requires at least 25 voters. This bound was recently brought down to 21 voters by \citet{Kard14a}. Simplified proofs of Moulin's theorem
are given by \citet{Schu03a} and \citet{Smit07a}. 
\citet{Holz88a} and \citet{SaZw09a} strengthen Moulin's theorem by weakening Condorcet-consistency and participation, respectively.
\citet{JPG09a} prove variants of Moulin's theorem for set-valued voting rules based on the optimistic and the pessimistic preference extension. 
\citet{Pere01a} defines a weaker notion of participation in the context of set-valued voting rules and shows that all common Condorcet extensions except the maximin rule and Young's rule violate this property. \citeauthor{Pere01a} notes that ``a practical question, which has not been dealt with here, refers to the number of candidates and voters that are necessary to invoke the studied paradoxes'' \citep[][p.~614]{Pere01a}. 

When assuming that voters have incomplete preferences over sets or lotteries, participation and Condorcet-consistency can be satisfies simultaneously and positive results for common Condorcet-consistent voting rules have been obtained by \citet{Bran11c} and \citet{BBGH15a,BBH15b,BBH15c}.
Abstention in slightly different contexts than the one studied in this paper recently caught the attention of computer scientists working on voting equilibria and campaigning 
\citep{DeEl10a,BFLR12a}.

The computer-aided techniques in this paper are inspired by \citet{TaLi09a}, who reduced well-known impossibility results from social choice theory---such as Arrow's theorem---to finite instances, which can then be checked by a SAT solver. This methodology has been extended and applied to new problems by \citet{GeEn11a}, \citet{BrGe15a}, and \citet{BBGH15a}. 
The results obtained by computer-aided theorem proving have already found attention in the social choice community
\citep{ChSe14a}. More generally, SAT solvers have also proven to be quite effective for other problems in economics. 
A prominent example is the ongoing work by \citet{FNL15a} in which SAT solvers are used for the development and execution of the FCC's upcoming reverse spectrum auction. 
In some respects, our approach also bears some similarities to \emph{automated mechanism design} \citep[see, \eg][]{CoSa02e}, where desirable properties are encoded and mechanisms are computed to fit specific problem instances.

\section{Preliminaries}
\label{sec:preliminaries}

Let $A$ be a set of $m$ alternatives and $\mathcal{N}$ be a set of $n$ voters. 
Whether no-show paradoxes occur depends on the exact values of $m$ and $n$.
By $\electorates :=2^\mathcal{N}\setminus\{\emptyset\}$ we denote the set of \emph{electorates}, \ie non-empty subsets of $\mathcal{N}$. 
For many of our results, we will take $A=\{\a,\b,\c,\d\}$, and we use the labels $\x, \y$ for arbitrary elements of $A$. 

A \emph{(strict) preference relation} is a complete, antisymmetric, and transitive binary relation on~$A$. 
The preference relation of voter~$i$ is denoted by~$\pref_i$. 
The set of all preference relations over $A$ is denoted by~$\mathcal{R}$.
For brevity, we denote by $\a\b\c\d$ the preference relation $\a \mathrel{\pref_i} \b \mathrel{\pref_i} \c \mathrel{\pref_i} \d$, eliding the identity of voter $i$, and similarly for other preferences.

A \emph{preference profile} $R$ is a function from an electorate $N\in\electorates$ to the set of preference relations $\mathcal{R}$. 
The set of all preference profiles is thus given by $\mathcal{R}^\electorates$. 
For the sake of adding and deleting voters, we define for any preference profile $R\in\mathcal{R}^N$ with $(i,\pref_i)\in R$, and $j\in\mathcal{N}\setminus N$, ${\pref_j}\in \mathcal{R}$ 
\begin{align*} &R - i := R\setminus\{(i,\pref_i)\}\text{,}&& R + {(j,\pref_j)} := R\cup\{(j,\pref_j)\}\text{.}
\end{align*}
If the identity of the voter is clear or irrelevant we sometimes, in slight abuse of notation, refer to $R - i$ by $R - {\pref_i}$, and write $R + {\pref_j}$ instead of $R + {(j,\pref_j)}$. 
If $k$ voters with the same preferences $\pref_i$ are to be added or removed, we write $R + k \cdot {\pref_i}$ and $R - k \cdot {\pref_i}$, respectively.

The \emph{majority margin} of 
$R$ is the map $g_R\colon A\times A \rightarrow \mathbb{Z}$ with $g_R(\x,\y) = |\{i\in N\mid \x\mathrel{\pref_i} \y\}| - |\{i\in N\mid \y\mathrel{\pref_i} \x\}|$. 
The majority margin can be viewed as the adjacency matrix of a \emph{weighted tournament} $T_R$. We say that a preference profile $R$ \emph{induces} the weighted tournament $T_R$. 

An alternative $\x$ is called \emph{Condorcet winner} if it wins against any other alternative in a majority contest, \ie if $g_R(\x,\y)>0$ for all $\y\in A\setminus \{\x\}$. 
Conversely, an alternative  $\x$ is a \emph{Condorcet loser} if $g_R(\x,\y)<0$ for all $\y\in A\setminus \{\x\}$. 

Our central object of study are \emph{voting rules}, \ie functions that assigns every preference profile a socially preferred alternative. 
Thus, a voting rule is a function $f\colon\mathcal{R}^\electorates \rightarrow A$. 

In this paper, we study voting rules that satisfy \emph{Condorcet-consistency} and \emph{participation}.

\begin{definition}
	A \emph{Condorcet extension} is a voting rule that selects the Condorcet winner whenever it exists. Thus, $f$ is a Condorcet extension if for every preference profile $R$ that admits a Condorcet winner $\x$, we have $f(R) = \x$.  We say that $f$ is \emph{Condorcet-consistent}.
\end{definition}

\begin{definition}
	A voting rule $f$ satisfies \emph{participation} if all voters always weakly prefer voting to not voting, \ie if $f(R) \pref_i f(R - i)$ for all $R\in\mathcal{R}^N$ and $i\in N$ with $N\in\electorates$. 
\end{definition}

Equivalently, participation requires that for all preference profiles $R$ not including voter $j$, we have $f(R + {\pref_j}) \pref_j f(R)$.

\section{Maximin and Kemeny's Rule}
\label{sec:maximin-kemeny}

The proofs of both positive and negative results to come were obtained through automated techniques that we describe in~\Cref{sec:sat}. To become familiar with the kind of arguments produced in this way, we will now study a more restricted setting which is of independent interest.

Specifically, let us consider voting rules that select winners in accordance with the popular \textit{maximin} and \textit{Kemeny} rules. For a preference profile $R$, an alternative $\x$ is a \textit{maximin winner} if it maximizes $\min_{\y\in A\setminus\{\x\}} g_R(\x,\y)$; thus, $\x$ never gets defeated too badly in pairwise comparisons. An alternative $\x$ is a \textit{Kemeny winner} if it is ranked first in some \textit{Kemeny ranking}. A Kemeny ranking is a preference relation ${\pref_\mathrm{K}}\in\mathcal{R}$ maximizing agreement with voters' individual preferences, \ie it maximizes the quantity $\sum_{i\in N} |{\pref_\mathrm{K}} \cap {\pref_i}|$.

We call a voting rule a \textit{maximin extension} (resp.\ \textit{Kemeny extension}) if it always selects a maximin winner (resp.\ Kemeny winner). Since a Condorcet winner, if it exists, is always the unique maximin and Kemeny winner of a preference profile, any such voting rule is also a Condorcet extension. We can now prove an easy version of Moulin's theorem for these more restricted voting rules.

To this end, we first prove a useful lemma allowing us to extend impossibility proofs for $4$~alternatives to also apply if there are more than $4$~alternatives. Its proof gives a first hint on how Condorcet-consistency and participation interact.

\begin{lemma}
	\label{lem:avoid-last}
	Suppose that $f$ is a Condorcet extension satisfying participation.
	Let $R$ be a preference profile and $B\subsetneq A$ a set of \emph{bad} alternatives such that each voter ranks every $\x\in B$ below every $\y\in A\setminus B$.
	Then $f(R) \notin B$. 
\end{lemma}
\begin{proof}
	By induction on the number of voters $|N|$ in $R$. If $R$ consists of a single voter $i$, then, since $f$ is a Condorcet extension, $f(R)$ must return $i$'s top choice, which is not bad. If $R$ consists of at least 2 voters, and $i\in N$, then by participation $f(R) \pref_i f(R - i)$. If $f(R)$ were bad, then so would be $f(R-i)$, contradicting the inductive hypothesis.
\end{proof}

The following computer-aided proofs, just like the more complicated proofs to follow, can be understood solely by carefully examining the corresponding `proof diagram'. An arrow such as
\tikz[every node/.style={inner sep=2pt},
lbl/.style={fill=white, inner sep=1pt},
->,baseline=(a.base)]
{\draw node (a) {$R$} 
	node[right=1.75cm of a] (b) {$R'$} 
	(a) edge node[lbl,xshift=-0.1cm] {${}+ \yes{\a\b}\no{\c\d}$} (b);}
indicates that profile $R'$ is obtained from $R$ by adding a voter $\a\b\c\d$, and is read as ``if one of the bold green alternatives (here $\yes{\a\b}$) is selected at $R$, then one of them is selected at $R'$'' (by participation). A circled node 
\tikz[condorcet/.style={circle,draw, inner sep=1.5pt},baseline=(a.base)]
{\draw node[condorcet] (a) {$\a$};}
indicates a profile admitting Condorcet winner $\a$, although in the proofs of Theorems~\ref{thm:maximin} and~\ref{thm:kemeny}, we use it to refer to maximin and Kemeny winners, respectively.

\begin{theorem}
	\label{thm:maximin}
	There is no maximin extension that satisfies participation for $m\geq 4$ and $n\geq 7$. (For $m=4$ and $n\leq 6$, such a maximin extension exists.)
\end{theorem}
\begin{proof}
	Let $f$ be a maximin extension which satisfies participation. Consider the following 6-voter profile $R$:
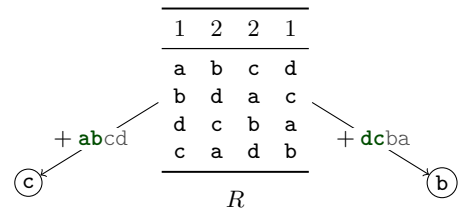
	\[\begin{tikzpicture}
	[->,
	level distance=17mm,
	sibling distance=55mm,
	level 2/.style={sibling distance=20mm},
	ne/.style={inner sep=2pt},
	empty/.style={circle,draw=black!75,fill=black!40, inner sep=1.5pt},
	condorcet/.style={draw, circle, inner sep=1.5pt},
	lbl1/.style={sloped, above},
	lbl/.style={fill=white, inner sep=1pt}
	]
	
	\draw node	[inner sep=5pt] 
	(fakeroot) at (0, 0) 
	{
		\begin{array}[t]{c@{\quad}c@{\quad}c@{\quad}c}
		\toprule
		1&2&2&1\\
		\midrule
		\a&\b&\c&\d\\
		\b&\d&\a&\c\\
		\d&\c&\b&\a\\
		\c&\a&\d&\b\\
		\bottomrule\\[-5pt]
		\multicolumn{4}{c}{R}
		\end{array}
	};
	\draw node [rectangle, minimum width=3cm, minimum height=1.25cm] 
	(root) at (0, 0.75cm) 
	{}
	child {
		node[condorcet] (alpha) {$\c$}
		edge from parent node[lbl,xshift=-0.15cm] {${}+ \yes{\a\b}\no{\c\d}$}
	}
	child {
		node[condorcet] (alphaprime) {$\b$}
		edge from parent node[lbl] {${}+ \yes{\d\c}\no{\b\a}$}
	}
	;
	\ifshowtournaments
	\begin{scope}[shift={(alpha.south)},shift={(-0.65cm,-0.5cm)}, scale=1.3,vertex/.style={fill=white,circle,draw,minimum size=13.9,inner sep=2}]
	\draw
	(0,0) node[vertex] (a) {$\a$}
	(1,0) node[vertex] (b) {$\b$} 
	(1,-1) node[vertex, very thick] (c) {$\c$} 
	(0,-1) node[vertex] (d) {$\d$};
	
	\draw [-latex] (b) to node [white,right] {1} (c);
	\draw [-latex] (c) to node [above,sloped,near start] {3} (a);
	\draw [-latex] (d) to node [white,above] {1} (c);
	\draw [-latex] (a) to node [above] {3} (b);
	\draw [-latex] (b) to node [sloped,above,near end] {5} (d);
	\draw [-latex] (a) to node [white, left] {1} (d);
	\end{scope}
	
	\begin{scope}[shift={(alphaprime.south)},shift={(-0.65cm,-0.5cm)}, scale=1.3,vertex/.style={fill=white,circle,draw,minimum size=13.9,inner sep=2}]
	\draw
	(0,0) node[vertex] (a) {$\a$}
	(1,0) node[vertex, very thick] (b) {$\b$} 
	(1,-1) node[vertex] (c) {$\c$} 
	(0,-1) node[vertex] (d) {$\d$};
	
	\draw [-latex] (c) to node [white,above] {1} (b);
	\draw [-latex] (c) to node [above,sloped,near start] {5} (a);
	\draw [-latex] (d) to node [below] {3} (c);
	\draw [-latex] (a) to node [white,above] {1} (b);
	\draw [-latex] (b) to node [sloped,above,near end] {3} (d);
	\draw [-latex] (d) to node [white, left] {1} (a);
	\end{scope}
	\fi
	\end{tikzpicture}
	\vspace{-10pt}
	\]
	
	Suppose $f(R)\in \{\a,\b\}$. 
	Adding an $\a\b\c\d$ vote leads to a weighted tournament in which alternative $\c$ is the unique maximin winner. 
	But this contradicts participation since the added voter would benefit from abstaining the election. 
	
	Symmetrically, if $f(R)\in \{\c,\d\}$, then adding a $\d\c\b\a$ vote leads to a weighted tournament in which $\b$ is the maximin winner, again contradicting participation. 
	The symmetry of the argument is due to an automorphism of $R$, namely the relabelling of alternatives according to $\a\b\c\d\mapsto\d\c\b\a$.
	
	If $m>4$, we add new bad alternatives $\x_1, \x_2, \dots, \x_{m-4}$ to the bottom of $R$ and of the additional voters. By \Cref{lem:avoid-last}, $f$ chooses from $\{\a,\b,\c,\d\}$ at each step, completing the proof. 
	
	The existence result for $n \leq 6$ is established by the methods described in \Cref{sec:sat}.
\end{proof}

For 3~alternatives, \citet{Moul88b} proved that the voting rule that chooses the lexicographically first maximin winner satisfies participation. \Cref{thm:maximin} shows that this is not the case for $4$~alternatives, even if there are only $7$~voters and no matter how we pick among maximin winners.

\begin{theorem}
	\label{thm:kemeny}
	There is no Kemeny extension that satisfies participation for $m\geq 4$ and $n\geq 4$. (For $m=4$ and $n\leq 3$, such a Kemeny extension exists.)
\end{theorem}
\begin{proof}
	Let $f$ be a Kemeny extension which satisfies participation. Consider the following 4-voter profile $R$:
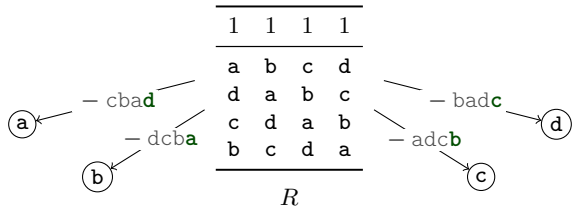
	\[\begin{tikzpicture}
	[->,
	level distance=19mm,
	sibling distance=21mm,
	level 2/.style={sibling distance=20mm},
	ne/.style={inner sep=2pt},
	empty/.style={circle,draw=black!75,fill=black!40, inner sep=1.5pt},
	condorcet/.style={circle,draw, inner sep=1.5pt},
	lbl/.style={fill=white, inner sep=1pt}
	]
	
	\draw node	[inner sep=5pt] 
	(fakeroot) at (0, 0) 
	{
		\begin{array}[t]{c@{\quad}c@{\quad}c@{\quad}c}
		\toprule
		1&1&1&1\\
		\midrule
		\a&\b&\c&\d\\
		\d&\a&\b&\c\\
		\c&\d&\a&\b\\
		\b&\c&\d&\a\\
		\bottomrule\\[-5pt]
		\multicolumn{4}{c}{R}
		\end{array}
	};
	\draw node [rectangle, minimum width=2.5cm, minimum height=1.4cm] 
	(root) at (0, 0.7cm) 
	{}
	\ifshowtournaments
	child {
		node[condorcet] (1) {$\a$}
		edge from parent node[lbl] {${}- \no{\c\b\a}\yes{\d}$}
	}
	child {
		node[condorcet] (2) {$\b$}
		edge from parent node[lbl] {${}- \no{\d\c\b}\yes{\a}$}
	}
	child {
		node[condorcet] (3) {$\c$}
		edge from parent node[lbl] {${}- \no{\a\d\c}\yes{\b}$}
	}
	child {
		node[condorcet] (4) {$\d$}
		edge from parent node[lbl] {${}- \no{\b\a\d}\yes{\c}$}
	}
	;
	\begin{scope}[shift={(1.south)},shift={(-0.65cm,-0.5cm)}, scale=1.2,vertex/.style={fill=white,circle,draw,minimum size=13.9,inner sep=2}]
	\draw
	(0,0) node[vertex, very thick] (a) {$\a$}
	(1,0) node[vertex] (b) {$\b$} 
	(1,-1) node[vertex] (c) {$\c$} 
	(0,-1) node[vertex] (d) {$\d$};
	
	\draw [-latex] (c) to node [white,right] {1} (b);
	\draw [-latex] (d) to node [above] {3} (c);
	\draw [-latex] (b) to node [white,above] {1} (a);
	\draw [-latex] (a) to node [white, left] {1} (d);

	\draw [-latex] (a) to node [white,above,sloped,near start] {1} (c);
	\draw [-latex] (d) to node [white,sloped,above,near end] {1} (b);
	\end{scope}
	\begin{scope}[shift={(2.south)},shift={(-0.65cm,-0.5cm)}, scale=1.2,vertex/.style={fill=white,circle,draw,minimum size=13.9,inner sep=2}]
	\draw
	(0,0) node[vertex] (a) {$\a$}
	(1,0) node[vertex, very thick] (b) {$\b$} 
	(1,-1) node[vertex] (c) {$\c$} 
	(0,-1) node[vertex] (d) {$\d$};
	
	\draw [-latex] (c) to node [white,right] {1} (b);
	\draw [-latex] (d) to node [white,above] {1} (c);
	\draw [-latex] (b) to node [white,above] {1} (a);
	\draw [-latex] (a) to node [left] {3} (d);

	\draw [-latex] (a) to node [white,above,sloped,near start] {1} (c);
	\draw [-latex] (b) to node [white,sloped,above,near end] {1} (d);
	\end{scope}
	\begin{scope}[shift={(3.south)},shift={(-0.65cm,-0.5cm)}, scale=1.2,vertex/.style={fill=white,circle,draw,minimum size=13.9,inner sep=2}]
	\draw
	(0,0) node[vertex] (a) {$\a$}
	(1,0) node[vertex] (b) {$\b$} 
	(1,-1) node[vertex, very thick] (c) {$\c$} 
	(0,-1) node[vertex] (d) {$\d$};
	
	\draw [-latex] (c) to node [white,right] {1} (b);
	\draw [-latex] (d) to node [white,above] {1} (c);
	\draw [-latex] (b) to node [above] {3} (a);
	\draw [-latex] (a) to node [white, left] {1} (d);

	\draw [-latex] (c) to node [white,above,sloped,near start] {1} (a);
	\draw [-latex] (b) to node [white,sloped,above,near end] {1} (d);
	\end{scope}
	\begin{scope}[shift={(4.south)},shift={(-0.65cm,-0.5cm)}, scale=1.2,vertex/.style={fill=white,circle,draw,minimum size=13.9,inner sep=2}]
	\draw
	(0,0) node[vertex] (a) {$\a$}
	(1,0) node[vertex] (b) {$\b$} 
	(1,-1) node[vertex] (c) {$\c$} 
	(0,-1) node[vertex, very thick] (d) {$\d$};
	
	\draw [-latex] (c) to node [right] {3} (b);
	\draw [-latex] (d) to node [white,above] {1} (c);
	\draw [-latex] (b) to node [white,above] {1} (a);
	\draw [-latex] (a) to node [white, left] {1} (d);

	\draw [-latex] (c) to node [white,above,sloped,near start] {1} (a);
	\draw [-latex] (d) to node [white,sloped,above,near end] {1} (b);
	\end{scope}
	\else
	child {
		node[condorcet] (1) at (-0.4cm, 1cm) {$\a$}
		edge from parent node[lbl] {${}- \no{\c\b\a}\yes{\d}$}
	}
	child {
		node[condorcet] (2) at (-1.5cm, 0.3cm) {$\b$}
		edge from parent node[lbl] {${}- \no{\d\c\b}\yes{\a}$}
	}
	child {
		node[condorcet] (3) at (1.5cm, 0.3cm) {$\c$}
		edge from parent node[lbl] {${}- \no{\a\d\c}\yes{\b}$}
	}
	child {
		node[condorcet] (4) at (0.4cm, 1cm) {$\d$}
		edge from parent node[lbl] {${}- \no{\b\a\d}\yes{\c}$}
	}
	;
	\fi
	\end{tikzpicture}
	\vspace{-10pt}
	\]
	
	Suppose $f(R)=\d$. 
	Then removing $\c\b\a\d$ from $R$ yields a weighted tournament in which $\a$ is the (unique) Kemeny winner, which contradicts participation. 
	Analogously, we can exclude the other three possible outcomes for $R$ by letting a voter abstain, which always leads to a unique Kemeny winner and a contradiction with participation. The arguments are identical because $R$ is completely symmetric in the sense that for any pair of alternatives~$\x$ and $\y$, there is an automorphism of $R$ that maps $\x$ to $\y$.
	
	Just like for \Cref{thm:maximin}, if $m>4$, we add new bad alternatives $\x_1, \x_2, \dots, \x_{m-4}$ to the bottom of $R$ and of the additional voters. By \Cref{lem:avoid-last}, $f$ chooses from $\{\a,\b,\c,\d\}$ at each step, completing the proof.
\end{proof}

One remarkable and unexpected aspect of the computer-aided proofs above is that their simplicity is due to automorphisms of the underlying preference profiles. Similar automorphisms will also be used in the proofs of the stronger theorems in Sections \ref{sec:results}, \ref{sec:set-valued}, and \ref{sec:probabilistic}. We emphasize that these symmetries are not hard-coded into our problem specification and, to the best of our knowledge, have not been exploited in previous proofs of similar statements.

\section{Method: SAT Solving for \\ Computer-Aided Proofs}\label{sec:sat}

The bounds in this paper were obtained with the aid of a computer. In this section, we describe the method that we employed. The main tool in our approach is an encoding of our problems into propositional logic. We then use SAT solvers to decide whether (in a chosen setting) there exists a Condorcet extension satisfying participation. If the answer is yes, the solver returns an explicit such voting rule. If the answer is no, we use a process called \textit{MUS extraction} to find a short certificate of this fact which can be translated into a human-readable proof. By successively proving stronger theorems and using the insights obtained through MUS extraction, we arrived at results as presented in their full generality in this paper. 

\subsection{SAT Encoding}
\label{subsec:sat-encoding}

``For $n$~voters and $4$~alternatives, is there a voting rule~$f$ that satisfies Condorcet-consistency and participation?'' 

A natural encoding of this question into propositional logic proceeds like this: Generate all profiles over 4 alternatives with at most $n$ voters. For each such profile $R$, introduce 4 propositional variables $x_{R,\a}$, $x_{R,\b}$, $x_{R,\c}$, $x_{R,\d}$, where the intended meaning of $x_{R,\a}$ is
\[ x_{R,\a} \text{ is set true} \iff f(R) = \a. \]
We then add clauses requiring that for each profile $R$, $f(R)$ takes \textit{exactly one} value, and we add clauses requiring $f$ to be Condorcet-consistent and satisfy participation.

Sadly, the encoding sketched above is not tractable for the values of $n$ that we are interested in: the number of variables and clauses used grows as $\Theta(24^n)$, because there are $4! = 24$ possible preference relations over 4 alternatives and thus $24^n$ profiles with $n$ voters. 
For $n=7$, this leads to more than 400 billion variables, and for $n=15$ we exceed $10^{22}$ variables.

To escape this combinatorial explosion, we will temporarily restrict our attention to \textit{pairwise} voting rules. This means that we assign an outcome alternative $f(T)$ to every weighted tournament $T$. We then define a voting rule that assigns  the outcome $f(T_R)$ to each preference profile $R$, where $T_R$ is the weighted tournament induced by $R$.

The number of tournaments induced by profiles with $n$ voters grows much slower than the number of profiles---our computer enumeration
suggests a growth of order about $1.5^n$. This much more manageable (yet still exponential) growth allows us to consider problem instances up to $n \approx 16$ which turns out to be \emph{just} enough.

Other than referring to (weighted) tournaments instead of profiles, our encoding into logic now proceeds exactly like before. For each tournament $T$, we introduce the variables $x_{T,\a}$, $x_{T,\b}$, $x_{T,\c}$, $x_{T,\d}$ and define the formulas
{\setlength{\belowdisplayskip}{2pt}
\begin{align*}
\formulaname{non-empty}_T &:= x_{T,\a} \lor x_{T,\b} \lor x_{T,\c} \lor x_{T,\d} \\
\formulaname{mutex}_T &:= \bigland_{\x\neq\y}  (\lnot x_{T,\x} \lor \lnot x_{T,\y})
\end{align*}}
With our intended interpretation of the variables $x_{T,\x}$, all models of 
$\bigland_T \formulaname{non-empty}_T \land \formulaname{mutex}_T$
are functions from tournaments into $\{\a,\b,\c,\d\}$. (The word $\formulaname{mutex}$ abbreviates `mutual exclusion' and corresponds to the voting rule selecting a unique winner.)

Since we are interested in voting rules that satisfy participation, we also need to encode this property. To this end, let $T = T_R$ be a tournament induced by $R$ and let $\pref$ be a preference relation. Define $T + {\pref} := T_{R+{\pref}}$. (The tournament $T+{\pref}$ is independent of the choice of $R$.) We define
\[ 
\formulaname{participation}_{T,{\pref}} := \bigland_\x \bigg( x_{T,\x} \to \biglor_{\y \pref \x} x_{T + {\pref}, \y} \bigg). \]
Requiring $f$ to be Condorcet-consistent is straightforward: if tournament $T$ admits $\b$ as the Condorcet winner, we add
\[ \formulaname{condorcet}_T := \lnot x_{T,\a} \land x_{T,\b} \land \lnot x_{T,\c} \land \lnot x_{T,\d}, \]
and we add similar formulas for each tournament that admits a Condorcet winner.
Then the models of the conjunction of all the $\formulaname{non-empty}$, $\formulaname{mutex}$, $\formulaname{participation}$, and $\formulaname{condorcet}$ formulas are precisely the pairwise voting rules satisfying Condorcet-consistency and participation.

By adapting the $\formulaname{condorcet}$ formulas, we can impose more stringent conditions on $f$---this is how our results for maximin and Kemeny extensions are obtained. We can also use this to exclude Pareto-dominated alternatives, and to require $f$ to always pick from the top cycle.

For some purposes it will be useful not to include the $\formulaname{mutex}$ clauses in our final formula. Models of this formula then correspond to
\textit{set-valued}
voting rules that satisfy participation interpreted according to the optimistic preference extension. See \Cref{sec:set-valued} for results in this setting.

\subsection{SAT Solving and MUS Extraction}
The formulas we have obtained above are all given in \textit{conjunctive normal form} (CNF), and thus can be evaluated without further transformations by any off-the-shelf SAT solver.
In order to physically produce a CNF formula as described, we employ a straightforward Python script that performs a breadth-first search to discover all weighted tournaments with up to $n$ voters (see \Cref{alg:write_formula} for a schematic overview of the program). The script outputs a CNF formula in the standard DIMACS format,
and also outputs a file that, for each variable $x_{T,\x}$, records the tournament $T$ and alternative $\x$ it denotes. This is necessary because the DIMACS format uses uninformative variable descriptors (consecutive integers) and memorizing variable meanings allows us to interpret the output of the SAT solver.

\begin{algorithm}
	\caption{Generate formula for up to $n$ voters}
	\begin{algorithmic}
		\State $T_0 \gets \{$weighted tournament on $\{\a,\b,\c,\d\}$ with
		\State \qquad \quad all edges having weight 0$\}$.
		\For{$k = 1,\dots,n$}
			\State $T_k \gets \emptyset$
			\For{$T \in T_{k-1}$}
				\For{${\pref} \in \mathcal R$ }
					\State Calculate $T' := T + {\pref}$
					\State Add $T'$ to $T_k$
					\State Write $\formulaname{non-empty}_{T'}, \formulaname{mutex}_{T'}, \formulaname{condorcet}_{T'}$
					\State Write $\formulaname{participation}_{T,{\pref}}$
				\EndFor
			\EndFor
		\EndFor
	\end{algorithmic}
	\label{alg:write_formula}
\end{algorithm}

As an example, the output formula for $n = 15$ in DIMACS format has a size of about 7~GB and uses 
50 million variables and 
2 billion clauses, taking 6.5 hours to write. Plingeling~\citep{Bier13a}, a popular SAT solver that we used for all results in this paper, solves this formula in 50 minutes of wall clock time, half of which is spent parsing the formula.

In case a given instance is satisfiable, the solver returns a satisfying assignment, giving us an existence proof and a concrete example for a voting rule satisfying participation (and any further requirements imposed). In case a given instance in unsatisfiable, we would like to have short certificate of this fact as well. One possibility for this is having the SAT solver output a resolution proof (in DRUP format, say). 
This yields a machine-checkable proof, but has two major drawbacks: the generated proofs can be uncomfortably large \citep{KoLi14a},
and they do not give human-readable insights about \textit{why} the formula is unsatisfiable.

We handle this problem by computing a \textit{minimal unsatisfiable subset (MUS)} of the unsatisfiable CNF formula. An MUS is a subset of the clauses of the original formula which itself is unsatisfiable, and is minimally so: removing any clause from it yields a satisfiable formula. We used the tools MUSer2 \citep{BeMa12a} and MARCO \citep{LPMM15a} to extract MUSes. 
If an unsatisfiable formula admits a very small MUS, it is often possible to obtain a human-readable proof of unsatisfiability from it \citep{BrGe15a,BBGH15a}. 

Note that for purposes of extracting human-readable proofs, it is desirable for the MUS to be as small as possible, and also to refer to as few different tournaments as possible. The first issue can be addressed by running the MUS extractor repeatedly, instructing it to order clauses randomly (note that clause sets of different cardinalities can be minimally unsatisfiable with respect to set inclusion); similarly, we can use MARCO to enumerate all MUSes and look for small ones.
The second issue can be addressed by computing a \textit{group MUS}: here, we partition the clauses of the CNF formula into \textit{groups}, and we are looking for a minimal set of groups that together are unsatisfiable. In our case, the clauses referring to a given tournament $T$ form a group. In practice, finding a group MUS first and then finding a standard (clause-level) MUS within the group MUS yielded sets of size much smaller than MUSes returned without the intermediate group-step (often by a factor of 10).

To translate an MUS into a human-readable proof, we created another program that visualized the MUS in a convenient form.\footnote{Roughly, the visualization program proceeds by drawing an edge for every $\formulaname{participation}_{T,{\pref}}$ clause that occurs in the MUS, and marks the nodes for which $\formulaname{condorcet}_{T}$ clauses appear in the MUS.} 
Indeed, this program outputs the `proof diagrams' like \Cref{fig:tree12} that appear throughout this paper (though we re-typeset them). We think that interpreting these diagrams is quite natural (and is perhaps even easier than reading a textual translation). More importantly, the automatically produced graphs allowed us to quickly judge the quality of an extracted MUS.

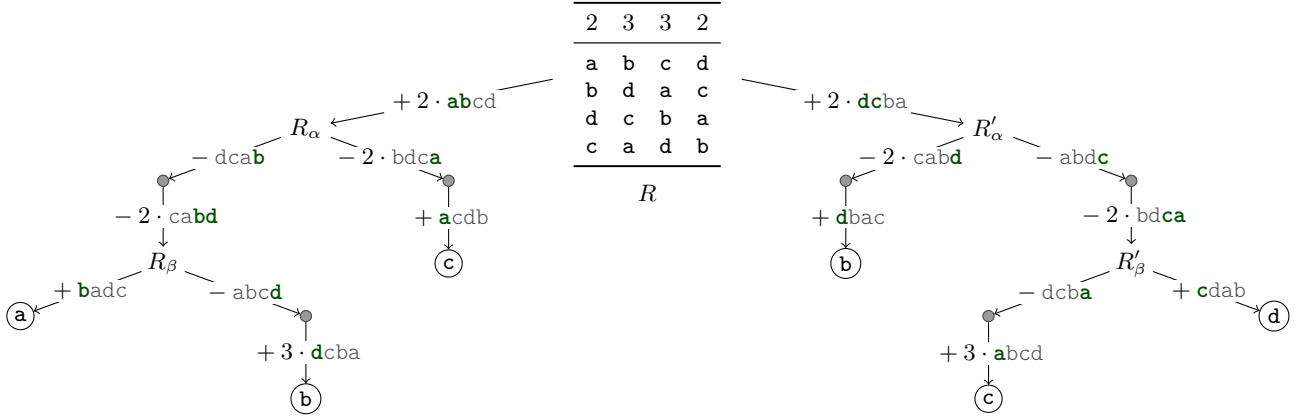
\begin{figure*}[t]
	\centering
	\begin{tikzpicture}
	[->,
	level distance=9mm,
	sibling distance=55mm,
	level 2/.style={level distance=7mm,sibling distance=23mm},
	level 3/.style={level distance=11mm},
	level 4/.style={level distance=7mm},
	level 5/.style={level distance=11mm},
	xscale=1.65,
	ne/.style={inner sep=2pt},
	empty/.style={circle,draw=black!75,fill=black!40, inner sep=1.5pt},
	condorcet/.style={circle,draw, inner sep=1.5pt},
	lbl/.style={fill=white, inner sep=1pt}
	]
	
	\draw node [inner sep=5pt] 
	(fake root) at (0, 0) 
	{
		\begin{array}[t]{c@{\quad}c@{\quad}c@{\quad}c}
		\toprule
		2&3&3&2\\
		\midrule
		\a&\b&\c&\d\\
		\b&\d&\a&\c\\
		\d&\c&\b&\a\\
		\c&\a&\d&\b\\
		\bottomrule\\[-5pt]
		\multicolumn{4}{c}{R}
		\end{array}
	};
	\draw node [rectangle, minimum width=2.5cm, minimum height=1cm] 
	(root) at (0, 0.6cm) 
	{}
	child {
		node (alpha) {$R_\alpha$}
		child {
			node[empty] (1) {}
			child {
				node[ne] (beta) {$R_\beta$}
				child {
					node[condorcet] (ca) {$\a$}
					edge from parent node[lbl] {${}+ \yes{\b}\no{\a\d\c}$}
				}
				child {
					node[empty] (2) {}
					child {
						node[condorcet] (cb) {$\b$}
						edge from parent node[lbl] {${}+3\cdot \yes{\d}\no{\c\b\a}$}
					}
					edge from parent node[lbl] {${}- \no{\a\b\c}\yes{\d}$}
				}
				edge from parent node[lbl] {${}-2\cdot \no{\c\a}\yes{\b\d}$}
			}
			edge from parent node[lbl] {${}- \no{\d\c\a}\yes{\b}$}
		}
		child {
			node[empty] (r2) {}
			child {
				node[condorcet] (cc) {$\c$}
				edge from parent node[lbl] {${}+ \yes{\a}\no{\c\d\b}$}
			}
			edge from parent node[lbl] {${}-2\cdot \no{\b\d\c}\yes{\a}$}
		}
		edge from parent 
		node[lbl] {${}+2\cdot \yes{\a\b}\no{\c\d}$}
	}
	child {
		node (alphaprime) {$R_\alpha'$}
		child {
			node[empty] (r4) {}
			child {
				node[condorcet] (cbprime) {$\b$}
				edge from parent node[lbl] {${}+ \yes{\d}\no{\b\a\c}$}
			}
			edge from parent node[lbl] {${}-2\cdot \no{\c\a\b}\yes{\d}$}
		}
		child {
			node[empty] (1prime) {}
			child {
				node[ne] (betaprime) {$R_\beta'$}
				child {
					node[empty] (2prime) {}
					child {
						node[condorcet] (cc) {$\c$}
						edge from parent node[lbl] {${}+3\cdot \yes{\a}\no{\b\c\d}$}
					}
					edge from parent node[lbl] {${}- \no{\d\c\b}\yes{\a}$}
				}
				child {
					node[condorcet] (cd) {$\d$}
					edge from parent node[lbl] {${}+ \yes{\c}\no{\d\a\b}$}
				}
				edge from parent node[lbl] {${}-2\cdot \no{\b\d}\yes{\c\a}$}
			}
			edge from parent node[lbl] {${}- \no{\a\b\d}\yes{\c}$}
		}
		edge from parent node[lbl] {${}+2\cdot \yes{\d\c}\no{\b\a}$}
	}
	;
	\end{tikzpicture}
	\caption{Computer-aided proof of \Cref{thm:anonymous} in graphical form, showing that there is no Condorcet extension that satisfies participation for $m\ge 4$ and $n\ge 12$. See \Cref{sec:maximin-kemeny} for an explanation of how to read this diagram.}
	\label{fig:tree12}
\end{figure*}

\subsection{Incremental Proof Discovery}

The SAT encoding described in \Cref{subsec:sat-encoding} only concerns pairwise voting rules, yet none of the (negative) results in this paper require or use this assumption. This is the product of multiple rounds of generating and evaluating SAT formulas, extracting MUSes, and using the insights generated by this as `educated guesses' to solve more general problems.

Following the process as described so far led to a proof that for $4$~alternatives and $12$~voters, there is no pairwise Condorcet extension that satisfies participation. That proof used the assumption of pairwiseness, \ie it assumed that the voting rule returns the same alternative on profiles inducing the same weighted tournament. However, intriguingly, the preference profiles mentioned in the proof did not contain all $4! = 24$ possible preference relations over $\{\a,\b,\c,\d\}$. In fact, it only used 10 of the possible orders.
Further, each profile included $R_0 = \{\a\b\d\c,\b\d\c\a,\c\a\b\d,\d\c\a\b\}$ as a subprofile. As we argued at the start of \Cref{subsec:sat-encoding}, it is intractable to search over the entire space of preference profiles.  On the other hand, it is much easier to merely search over all extensions of $R_0$ that contain at most $n = 12$ voters and only contain copies of the 10 orders previously identified. The SAT formula produced by doing exactly this turned out to be unsatisfiable, and a small MUS extracted from it gave rise to \Cref{thm:anonymous}.

The proof of \Cref{thm:optimist} for 17 voters was obtained by running \Cref{alg:write_formula} with $T_0$ initialized to the weighted tournament induced by the initial profile $R$ used in the proof of \Cref{thm:anonymous}. Before finding this tournament, we tried various other tournaments as $T_0$, including ones featuring in Moulin's original proof, and ones occurring at other steps in the proof of \Cref{thm:anonymous}, but $R$ turned out to give the best (and indeed a tight) bound, and additionally exhibits a lot of symmetry that was also present in the MUS we extracted.

\section{Main Result}\label{sec:results}

We are now in a position to state and prove our main claim that Condorcet extensions cannot avoid the no-show paradox for $12$~or more voters (when there are at least $4$~alternatives), and that this result is optimal.

\begin{theorem}
	\label{thm:anonymous}
	There is no Condorcet extension that satisfies participation for $m\ge 4$ and $n\ge 12$. 
\end{theorem}

\begin{proof}
	The proof follows the structure depicted in \Cref{fig:tree12}. Let $R$ be the preference profile shown there.
	
	Since $R$ remains fixed after relabelling alternatives according to $\a\b\c\d\mapsto\d\c\b\a$, we may assume without loss of generality that $f(R)\in\{\a,\b\}$. (An explicit proof in case $f(R)\in\{\c,\d\}$ is indicated in \Cref{fig:tree12}.) 
	
	By participation, it follows from $f(R)\in\{\a,\b\}$ that also $f(R_\alpha:=R + 2 \cdot \a\b\c\d)\in\{\a,\b\}$ since the voters with preferences $\a\b\c\d$ cannot be worse off by joining the electorate. 
	If $f(R_\alpha)=\a$, again by participation, removing $2$~voters
	with preferences $\b\d\c\a$ does not change the winning alternative (so $f(R_\alpha - 2 \cdot \b\d\c\a)=\a$), and neither does adding $\a\c\d\b$, so $f(R_\alpha - 2 \cdot \b\d\c\a + \a\c\d\b)=\a$, which, however, is in conflict with $R_\alpha - 2 \cdot \b\d\c\a + \a\c\d\b$ having a Condorcet winner, $\c$. 
	
	Thus we must have $f(R_\alpha)=\b$, which implies that $f(R_\alpha - \d\c\a\b)=\b$, and thus $f(R_\beta:=R_\alpha - \d\c\a\b - 2 \cdot \c\a\b\d)\in\{\b,\d\}$. 
	
	We again proceed by cases:
	If $f(R_\beta)=\b$, we can add a voter $\b\a\d\c$ to obtain a profile with Condorcet winner $\a$, which contradicts participation. 
	But then, if $f(R_\beta)=\d$, we get that $f(R_\beta - \a\b\c\d)=\d$ and, by another application of participation, that $f(R_\beta - \a\b\c\d + 3 \cdot \d\c\b\a)=\d$ in contrast to the existence of Condorcet winner $\b$, a contradiction.
	
	If $m>4$, we add bad alternatives $\x_1, \x_2, \dots, \x_{m-4}$ to the bottom of $R$ and all other voters. By \Cref{lem:avoid-last}, $f$ chooses from $\{\a,\b,\c,\d\}$ at each step, completing the proof.
\end{proof}

The following result establishes that our bound on the number of voters is tight. A very useful feature of our computer-aided approach is that we can easily add additional requirements for the desired voting rule. Two common requirements for voting rules are that they should only return alternatives that are \emph{Pareto-optimal} and contained in the \emph{top cycle} (also known as the \emph{Smith set}) \citep[see, \eg][]{Fish77a}.

\begin{theorem}
	\label{thm:positive-pairwise}
	There is a Condorcet extension $f$ that satisfies participation for $m=4$ and $n\leq 11$. Moreover, $f$ is pairwise, Pareto-optimal, and a refinement of the top cycle. 
\end{theorem}

The Condorcet extension $f$ is given as a look-up table, which is derived from the output of a SAT solver. The look-up table lists all $1,204,215$ weighted tournaments inducible by up to 11 voters and assigns each an output alternative (see \Cref{fig:lookup-table} for an excerpt). The relevant text file has a size of 28 MB (gzipped 4.5 MB). 

Comparing this voting rule with known voting rules, it turns out that it picks a maximin winner in 99.8\% and a Kemeny winner in 98\% of all weighted tournaments. Moreover, the rule agrees with the maximin rule with lexicographic tie-breaking on 95\% of weighted tournaments. The similarity with the maximin rule is interesting insofar as a well-documented flaw of the maximin rule is that it fails to be a refinement of the top cycle (and may even return Condorcet losers). Our computer-generated rule always picks from the top cycle while it remains very close to the maximin rule.

80\% of the considered weighted tournaments admit a Condorcet winner, which uniquely determines the output of the rule; this can be used to reduce the size of the look-up table.

\begin{figure}[t]
	\centering
	\begin{minipage}{0.47\columnwidth}
		{\small\ttfamily
			a,\#1,(1,1,1,1,1,1)\\
			a,\#1,(1,1,1,1,1,-1)\\
			a,\#1,(1,1,1,-1,1,1)\\
			a,\#1,(1,1,1,-1,-1,1)\\
			a,\#1,(1,1,1,1,-1,-1)\\
			a,\#1,(1,1,1,-1,-1,-1)\\
			b,\#1,(-1,1,1,1,1,1)\\
			b,\#1,(-1,1,1,1,1,-1)\\
			b,\#1,(-1,-1,1,1,1,1)\\
			b,\#1,(-1,-1,-1,1,1,1)\\
			{\color{black!70}b,\#1,(-1,1,-1,1,1,-1)}\\
			{\color{black!60}b,\#1,(-1,-1,-1,1,1,-1)}\\
			{\color{black!50}c,\#1,(1,-1,1,-1,1,1)}\\
			{\color{black!40}c,\#1,(1,-1,1,-1,-1,1)}
		}
	\end{minipage}
	\begin{minipage}{0.47\columnwidth}
		{\small\ttfamily
			{\color{black!40}a,\#11,(9,11,3,9,1,-9)}\\
			{\color{black!50}a,\#11,(11,9,3,7,1,-9)}\\
			{\color{black!60}c,\#11,(5,-9,-1,-11,-1,7)}\\
			{\color{black!70}c,\#11,(5,-9,-1,-11,-1,5)}\\
			c,\#11,(3,-11,-1,-9,1,7)\\
			c,\#11,(3,-11,-3,-9,1,7)\\
			c,\#11,(3,-11,-3,-11,-1,7)\\
			b,\#11,(-1,3,-5,-3,5,-3)\\
			b,\#11,(-3,3,-7,-3,5,-3)\\
			b,\#11,(-3,1,-7,-3,5,-3)\\
			c,\#11,(-3,1,-5,-5,5,-1)\\
			a,\#11,(3,7,11,-3,9,11)\\
			a,\#11,(3,7,11,-3,9,9)\\
			a,\#11,(3,7,11,-5,9,11)
		}
	\end{minipage}
	\caption{Excerpt of look-up table giving a pairwise Condorcet extension satisfying participation for $n\leq 11$ voters (from \Cref{thm:positive-pairwise}). Each row lists a weighted tournament as $(g_R(\a,\b),g_R(\a,\c),g_R(\a,\d),g_R(\b,\c),g_R(\b,\d),g_R(\c,\d))$ with a chosen alternative, and with the number of voters inducing the tournament.}
	\label{fig:lookup-table}
\end{figure}

\section{Set-valued voting rules}
\label{sec:set-valued}

A drawback of voting rules, as we defined them so far, is that that the requirement to always return a single winner is in conflict with basic fairness conditions, namely anonymity and neutrality. A large part of the social choice literature therefore deals with set-valued voting rules, where ties are usually assumed to be eventually broken by some tie-breaking mechanism.

A \emph{set-valued voting rule} (sometimes known as a voting \textit{correspondence} or as an \textit{irresolute} voting rule) is a function $F\colon\mathcal{R}^\electorates \rightarrow 2^A\setminus\{\emptyset\}$ that assigns each preference profile $R$ a non-empty set of alternatives. The function $F$ is a \emph{(set-valued) Condorcet extension} if for every preference profile $R$ that admits a Condorcet winner $\x$, we have $F(R) = \{\x\}$. 

Following the work of \citet{Pere01a} and \citet{JPG09a}, we seek to study the occurrence of the no-show paradox in this setting. To do so, we need to define appropriate notions of participation, and for this we will need to specify agents' preferences over \textit{sets} of alternatives. Here, we use the \textit{optimistic} and \textit{pessimistic preference extensions}. An optimist prefers sets with better most-preferred alternative, while a pessimist prefers sets with better least-preferred alternative. For example, if $U = \{\a,\b,\d\}$ and $V = \{\b,\c\}$, then an optimist with preferences $\a\b\c\d$ prefers $U$ to $V$, while a pessimist prefers $V$ to $U$. With these notions, we extend the participation property to set-valued voting rules.

\begin{definition}
	\label{def:set-valued-participation}
	A set-valued voting rule $F$ satisfies \emph{optimistic participation} if $\max_{\pref_i} F(R + {\pref_i}) \pref_i \max_{\pref_i} F(R)$. 
	
	A set-valued voting rule $F$ satisfies \emph{pessimistic participation} if $\min_{\pref_i} F(R) \pref_i \min_{\pref_i} F(R - i)$. 
\end{definition}

A set-valued voting rule $F$ is called \emph{resolute} if it always selects a single alternative, so that for all $R$ we have $|F(R)| = 1$. A (single-valued) voting rule $f$ is naturally identified with a resolute set-valued voting rule $F$; if $f$ satisfies participation, then this $F$ satisfies both optimistic and pessimistic participation. Hence, by \Cref{thm:positive-pairwise}, there is a (resolute) set-valued Condorcet extension $F$ that satisfies both optimistic and pessimistic participation. However, there might be hope that allowing voting rules to be irresolute also allows for participation to be attainable for more voters, and indeed this is the case.

\begin{theorem}
	\label{thm:positive-optimist}
	There is a set-valued Condorcet extension F that satisfies optimistic participation for $m = 4$ and $n\leq 16$, and also is Pareto-optimal and a refinement of the top cycle.
\end{theorem}

The SAT solver indicates that no such set-valued voting rule is pairwise. \Cref{thm:positive-optimist} is optimal in the sense that optimistic participation cannot be achieved if we allow for one more voter.

\begin{theorem}
	\label{thm:optimist}
	There is no set-valued Condorcet extension that satisfies optimistic participation for $m\ge 4$ and $n\ge 17$. 
\end{theorem}

\begin{proof}
	Let $F$ be such a function, and consider the following 10-voter profile $R$:
	\[\begin{tikzpicture}
	[->,
	level distance=14mm,
	sibling distance=33mm,
	level 2/.style={level distance=11mm, sibling distance=17mm},
	xscale=1.55,
	ne/.style={inner sep=2pt},
	empty/.style={circle,draw=black!75,fill=black!40, inner sep=1.5pt},
	condorcet/.style={circle,draw, inner sep=1.5pt},
	lbl/.style={fill=white, inner sep=1pt, font=\small}
	]
	
	\draw node	[inner sep=5pt, scale=0.9] 
	(fakeroot) at (0, 0) 
	{
		\begin{array}[t]{c@{\quad}c@{\quad}c@{\quad}c}
		\toprule
		2&3&3&2\\
		\midrule
		\a&\b&\c&\d\\
		\b&\d&\a&\c\\
		\d&\c&\b&\a\\
		\c&\a&\d&\b\\
		\bottomrule\\[-5pt]
		\multicolumn{4}{c}{R}
		\end{array}
	};
	\draw node [rectangle, minimum width=3cm, minimum height=1.25cm] 
	(root) at (0, 0.6cm) 
	{}
	child {
		node (alpha) {$R_\alpha$}
		child {
			node[condorcet] (1) {$\a$}
			edge from parent node[lbl,xshift=-0.2cm] {${}+5\cdot \yes{\b}\no{\a\c\d}$}
		}
		child {
			node[condorcet] (r2) {$\c$}
			edge from parent node[lbl] {${}+3\cdot \yes{\a}\no{\c\b\d}$}
		}
		edge from parent node[lbl,xshift=-0.3cm] {${}+2\cdot \yes{\a\b}\no{\c\d}$}
	}
	child {
		node (alphaprime) {$R_\alpha'$}
		child {
			node[condorcet] (r4) {$\b$}
			edge from parent node[lbl,xshift=-0.2cm] {${}+3\cdot \yes{\d}\no{\b\c\a}$}
		}
		child {
			node[condorcet] (1prime) {$\d$}
			edge from parent node[lbl] {${}+5\cdot \yes{\c}\no{\d\b\a}$}
		}
		edge from parent node[lbl, xshift=0.2cm] {${}+2\cdot \yes{\d\c}\no{\b\a}$}
	}
	;
	\end{tikzpicture}\]
	
	Suppose that either $\a\in F(R)$ or $\b\in F(R)$. (The case of $\c\in F(R)$ or $\d\in F(R)$ is symmetric.) Then let $R_\alpha := R + 2\cdot \a\b\c\d$. 
	By optimistic participation, we then have either $\a\in F(R_\alpha)$ or $\b\in F(R_\alpha)$. If we had $\a\in F(R_\alpha)$, then also $\a\in F(R_\alpha + 3\cdot \a\c\b\d)$ but this profile has Condorcet winner $\c$, and if $\b\in F(R_\alpha)$ then also $\b\in F(R_\alpha + 5\cdot \b\a\c\d)$ but this profile has Condorcet winner $\a$. This is a contradiction.
	
	This argument extends to more than $4$~alternatives by appealing to a set-valued analogue of \Cref{lem:avoid-last}.
\end{proof}

Inspecting Moulin's original proof \citep{Moul88b} shows that it also establishes an impossibility for optimistic participation (for $25$~voters). Apparently unaware of this, \citet{JPG09a} explicitly establish such a result for $27$~voters and $5$~alternatives. It is worth observing that each step of the proof of \Cref{thm:optimist} involves \textit{adding} voters to the current profile, and we never remove voters. In light of \Cref{def:set-valued-participation}, this is the reason why the proof establishes a result for optimistic participation. If we restrict ourselves to deleting voters, we obtain a result for pessimistic participation.

\begin{theorem}
	\label{thm:pessimist}
	There is no set-valued Condorcet extension that satisfies pessimistic participation for $m\ge 4$ and $n\ge 14$. On the other hand, for $m = 4$ and $n\leq 13$, there exists such a set-valued voting rule.
\end{theorem}

\begin{proof}[Sketch]
	The proof has a similar structure to the proof of \Cref{thm:anonymous}, displayed in \Cref{fig:tree12}. The initial profile of this proof is $R + 2\cdot\a\b\c\d + 2\cdot\d\c\b\a$, taking $R$ to be the profile of \Cref{fig:tree12}. We further replace proof steps in which voters are added by similar ones where voters are deleted, and invoke pessimistic participation at each such step to obtain a contradiction.
\end{proof}

This result strengthens a result of \citet{JPG09a}, who show that for $m \ge 5$ no set-valued Condorcet extension satisfying a property called ``weak translation invariance'' can also satisfy pessimistic participation. Our proof does not need the extra assumption, already works for $m = 4$ alternatives, and uses just 14 instead of 971 voters.\footnote{The large number of voters is due to several applications of the ``weak translation invariance'' axiom, each of which adds $5!=120$ voters to the preference profile under consideration.}

As previously observed, adding voters in our impossibility proofs corresponds to optimistic participation, while removing voters corresponds to pessimistic participation. In the proof of \Cref{thm:anonymous}, we use both operations, which allows for a tighter bound of just $12$~voters. In the set-valued setting, we can formulate this result in a slightly stronger way.

\begin{theorem}
	\label{thm:em-participation}
	There is no set-valued Condorcet extension that satisfies optimistic and pessimistic participation simultaneously for $m\ge 4$ and $n\ge 12$. On the other hand, for $m = 4$ and $n\leq 11$ such a set-valued rule exists (and also is Pareto-optimal and a refinement of the top cycle).
\end{theorem}

\begin{proof}
	Use the proof of \Cref{thm:anonymous}, invoking optimistic participation for edges labelled with the addition of a voter ($+$), and invoking pessimistic participation for edges labelled with removal of a voter ($-$). On the other hand, the (single-valued) voting rule of \Cref{thm:positive-pairwise} clearly satisfies both optimistic and pessimistic participation.
\end{proof}

The preference extension combining the optimistic and pessimistic preference extension is also known as the \textit{Egli-Milner extension}. 

\section{Probabilistic Voting Rules}
\label{sec:probabilistic}

A \textit{probabilistic voting rule} (also known as a \textit{social decision scheme}) assigns to each preference profile $R$ a probability distribution (or \textit{lottery}) over $A$. Thus, a probabilistic voting rule might assign a fair coin flip between $\a$ and $\b$ as the outcome of an election. 

Formally, let $\Delta(A) = \{ \mathbf p : A \to [0,1] : \sum_{\x\in A} \mathbf p (\x) = 1 \}$ be the set of lotteries over $A$; a lottery $\mathbf p \in \Delta(A)$ assigns probability $\mathbf p (\x)$ to alternative $\x$. A probabilistic voting rule $f$ is a function $f : \mathcal R^\electorates \to \Delta(A)$. In this context, we say that $f$ is a \textit{Condorcet extension} if $f(R)$ puts probability 1 on the Condorcet winner of $R$ whenever it exists: if $R$ admits $\x$ as the Condorcet winner, then $f(R)(\x) = 1$. 

As in the set-valued case, we need a notion of comparing outcomes in order to extend the definition of participation. Here, we use the concept of \textit{stochastic dominance (SD)}.

\begin{definition}
	Let ${\pref}\in\mathcal R$ be a preference relation over~$A$, and let $\mathbf p, \mathbf q \in \Delta(A)$ be lotteries. Then $\mathbf p$ is (weakly) SD-preferred over $\mathbf q$ by $\pref$ if for each alternative $\x$, we have
	\[ \textstyle \sum_{\y\pref\x} \mathbf p (y) \ge  \sum_{\y\pref\x} \mathbf q (y). \]
	In this case, we write $\mathbf p \pref^{\!\textup{SD}} \!\mathbf q$.
\end{definition}

For example, the lottery \smash{$\frac23\a + \frac13\c$} 
is SD-preferred to the lottery $\frac13\a+\frac13\b+\frac13\c$ by a voter with preferences $\a\b\c\d$. A voter with preferences $\b\a\c\d$ will feel the other way around. The main appeal of stochastic dominance 
stems from the following equivalence: 
$\mathbf p \pref^{\!\textup{SD}} \!\mathbf q$
if and only if $\mathbf p$ yields at least as much von-Neumann-Morgenstern utility as $\mathbf q$ 
under 
any utility function that is consistent with the ordinal preferences $\pref$.
Using this notion of comparing lotteries, we can define participation analogously to the previous settings.

\begin{definition}
	A probabilistic voting rule $f$ satisfies \emph{strong SD-participation} if $f(R) \pref_i^{\!\textup{SD}}\! f(R - i)$ for all $R\in\mathcal{R}^N$ and $i\in N$ with $N\in\electorates$. 
\end{definition}

Any (single-valued) voting rule $f$ can be seen as a probabilistic voting rule that puts probability 1 on its chosen alternative. If $f$ satisfies participation, then this derived probabilistic voting rule is easily seen to satisfy strong SD-participation. Hence \Cref{thm:positive-pairwise} gives us a probabilistic Condorcet extension that satisfies strong SD-participation for $n\leq 11$ voters and $m = 4$ alternatives.

We now establish a connection between strong SD-participation and the set-valued notions of participation that we considered in \Cref{sec:set-valued}. This connection will allow us to translate the impossibility results we obtained there to the probabilistic setting. To set up this connection, let us define the \textit{support} of  a lottery $\mathbf p \in \Delta(A)$ to be $\supp(\mathbf p) := \{ \x\in A : \mathbf p(\x) > 0 \}$.

\begin{proposition}
	Let $f$ be a probabilistic voting rule satisfying strong SD-participation. Let $F = \supp \circ f$ be the support of $f$, \ie $F(R) = \supp(f(R))$ for all profiles $R$. Then $F$ satisfies both optimistic and pessimistic participation.
\end{proposition}

\begin{proof}
	We only verify optimistic participation; the pessimistic case is similar. Let $R$ be a preference profile with electorate $N \in \electorates$, and let $i\in\mathcal N \setminus N$ be a voter with preferences $\pref_i$. Let $\x = \max_{\pref_i} F(R)$, and let $U = \{\y: \y\pref_i\x\}$. We need to show that $\max_{\pref_i} F(R + {\pref_i}) \pref_i \x$, by finding an alternative $\y\in U$ that is in the support of $f(R + {\pref_i})$.
	
	But since $f$ satisfies strong SD-participation, we have 
	\[\textstyle \sum_{\y\in U} f(R + {\pref_i})(\y) \ge \sum_{\y\in U} f(R)(\y) > 0,\] 
	where the strict inequality follows from the definition of the support and of $\x$. Hence some alternative from $U$ is in the support of $f(R + {\pref_i})$, as required.
\end{proof}

Putting these results together with the impossibility result of \Cref{thm:em-participation}, we obtain the following.

\begin{theorem}
	\label{thm:sd-participation}
	There is no probabilistic Condorcet extension that satisfies strong SD-participation for $n\ge 12$ and $m\ge 4$. On the other hand, for $m = 4$ and $n\leq 11$, such a probabilistic voting rule exists. 
\end{theorem}

\thmref{thm:sd-participation} resolves an open problem mentioned by \citet[][Sec.~6]{BBH15b}.

\section{Conclusions and Future Work}

We have given tight results delineating in which situations no-show paradoxes must occur. 
As such, our results nicely complement recent advances to satisfy Condorcet-consistency and participation by exploiting uncertainties of the voters about their preferences or about the voting rule's tie-breaking mechanism \citep{BBGH15a,BBH15b,BBH15c}. 

Due to unmanageable branching factors when there are $5$~alternatives (and hence $5!=120$ possible preference relations), we were unable to check using our approach whether no-show paradoxes occur with even less voters when the number of alternatives grows.
It would be interesting to gain a deeper understanding of the computer-generated Condorcet extension that satisfies participation for up to $11$~voters. So far, we only know that it (slightly) differs from all Condorcet extensions that are usually considered in the literature. As a first step, it would be desirable to obtain a representation of this rule that is more concise than a look-up table.

Another interesting topic for future research is to find optimal bounds for a variant of the no-show paradox due to \citet{SaZw09a}, in which participation is weakened to half-way monotonicity. 
Their proof requires 
$702$ voters.

\section*{Acknowledgments}
Christian Geist is supported by Deutsche Forschungsgemeinschaft under grant {BR~2312/9-1}. Dominik Peters is supported by EPSRC. Part of this work was conducted while Dominik Peters visited TUM, supported by the COST Action IC1205 on Computational Social Choice. 

\bibliographystyle{abbrv}

\end{document}